\documentclass[12pt,draftcls, onecolumn,journal]{IEEEtran}
\usepackage{times}
\usepackage{epstopdf}
\usepackage[reqno]{amsmath}
\usepackage{amsfonts}
\usepackage{caption2}
\usepackage{times,amsmath,epsfig}
\usepackage{latexsym,amssymb}
\usepackage{rotating,color}
\usepackage{cite}
\usepackage{extarrows}

\def\myQED{\mbox{\rule[0pt]{1.5ex}{1.5ex}}}

\DeclareMathOperator{\esssup}{esssup}
\DeclareMathOperator{\essinf}{essinf}

\newtheorem{thm}{Theorem}[section]

\newtheorem{prop}[thm]{Proposition}
\newtheorem{lem}[thm]{Lemma}

\newtheorem{rmk}[thm]{Remark}

\newcommand{\no}{\nonumber}
\IEEEoverridecommandlockouts
\makeatletter
\def\ps@headings{%
\def\@oddhead{\mbox{}\scriptsize\rightmark \hfil \thepage}%
\def\@evenhead{\scriptsize\thepage \hfil \leftmark\mbox{}}%
\def\@oddfoot{}%
\def\@evenfoot{}}
\makeatother
\begin{document}
\title{Non-Bayesian Quickest Detection with Stochastic Sample Right Constraints}
\author{Jun Geng and Lifeng Lai \thanks{This work was supported by the National Science Foundation CAREER award under grant CCF-10-54338 and by the National Science Foundation under grant DMS-12-65663.} \\
 Department of Electrical and Computer Engineering\\
 Worcester Polytechnic Institute\\
 Worcester, MA 01609, USA \\ Email:\{jgeng, llai\}@wpi.edu}
%\date{}
\maketitle %\pagestyle{plain}

%%%%%%%%%%%%%%% Article Body %%%%%%%%%%%%%%%%%%%%%%%%%%%%%%%%%%%%%%%%%

%%%%=======================================================

\begin{abstract}
In this paper, we study the design and analysis of optimal detection scheme for sensors that are deployed to monitor the change in the environment and are powered by the energy harvested from the environment. In this type of applications, detection delay is of paramount importance. We model this problem as quickest change detection problem with a stochastic energy constraint. In particular, a wireless sensor powered by renewable energy takes observations from a random sequence, whose distribution will change at a certain unknown time. Such a change implies events of interest. The energy in the sensor is consumed by taking observations and is replenished randomly. The sensor cannot take observations if there is no energy left in the battery. Our goal is to design a power allocation scheme and a detection strategy to minimize the worst case detection delay, which is the difference between the time when an alarm is raised and the time when the change occurs. Two types of average run length (ARL) constraint, namely an algorithm level ARL constraint and an system level ARL constraint, are considered. We propose a low complexity scheme in which the energy allocation rule is to spend energy to take observations as long as the battery is not empty and the detection scheme is the Cumulative Sum test. We show that this scheme is optimal for the formulation with the algorithm level ARL constraint and is asymptotically optimal for the formulations with the system level ARL constraint.
\end{abstract}


%%%%%%%%%%%%%%% Article Body %%%%%%%%%%%%%%%%%%%%%%%%%%%%%%%%%%%%%%%%%

%%%%=======================================================
\section{Introduction} \label{sec:intro}

Recently, the study of sensor networks powered by renewable energy harvested from the environment has attracted considerable attention~\cite{Mao:TAC:12,Ho:TSP:XX,Huang:JSAC:12,Yang:JSAC:11,Yang:JCN:12}. Compared with the sensor networks powered by batteries, the sensor networks powered by renewable energy have several unique features such as unlimited life span and high dependence on the environment etc. Optimal power management schemes for each individual sensor and scheduling protocols for the whole network have been developed to maximize utility functions of {\em communication related metrics} such as channel capacity, transmission delay or network throughput.
However, besides these communication related metrics, there are other {\em signal processing related performance metrics} that are also important for sensor networks targeted for certain applications. For example, if a sensor network is deployed to monitor the health of a bridge, then the detection delay between the time when a structural problem occurs and the time when an alarm is raised is of interest. As another example, if a sensor network is deployed for intruder detection, then the detection delay and the false alarm probability are of interest.
%However, the researches on {\em signal processing related performance metrics} for renewable energy powered sensors have not been investigated. Detection delay, which is one of such performance metrics, is important for sensor networks in many applications. For example, if a sensor network is deployed to monitor the health of a bridge, then the detection delay between the time when a structural problem occurs and the time when an alarm is raised is of interest.

Until now, these alternative but important performance metrics have not been investigated for sensors powered renewable energy. In this paper, we focus on the design of optimal power management schemes for such wireless sensor networks when the detection delay is of interest. In particular, we focus on so called ``quickest detection'' problem. In the quickest detection problem, wireless sensors are deployed to quickly detect the change (these terms will be precisely defined in the sequel) in the environment. Such changes typically imply certain activities of interest. For example, in the bridge monitoring, a change may imply that a certain structural problem has occurred in the bridge. As the result, it is of paramount importance to minimize the detection delay after the presence of such a change, hence the name of quickest detection.
%In this paper, we focus on the design of change detection schemes for renewable energy powered wireless sensors when detection delay is of interest. In particular, we focus on so called ``quickest detection'' problem. In quickest detection problem, wireless sensors are deployed to identify the the changes in the environment as quickly and accurately as possible, since the changes typically imply certain events of interest. For example, in bridge health monitoring, such a change may imply that a certain structural problem has occurred.
Besides this application, quickest detection also has many other potential applications, such as the quality control \cite{Roberts:Tech:66}, network intrusion detection \cite{Premkumar:INFOC:08}, cognitive radio \cite{Lai:GLOBE:08}, etc. We note that the detection delay in the change point detection problem refers to the delay between the time when a change occurs and the time when an alarm is raised. It is not the delay from time zero to the time when an alarm is raised, since we are interested in the change.

%One can formulate the task at the sensor as the quickest detection of abrupt changes in the stochastic process problem. % has many potential applications, such as the quality control \cite{Roberts:Tech:66}, network intrusion detection \cite{Premkumar:INFOC:08}, cognitive radio \cite{Lai:GLOBE:08}, etc.
Non-Bayesian quickest detection is one of the most important formulations, which was first studied by G. Lordon \cite{Lorden:AmS:71} and M. Pollak \cite{Pollak:AnS:85}. Under the non-Bayesian setup, a sensor sequentially observes a random sequence $\{X_{k}, k=1,2,\ldots\}$ with a fixed but unknown change point $t$. Before the change point $t$, the sequence $X_{1}, \ldots, X_{t-1}$ are independent and identically distributed (i.i.d.) with probability density function (pdf) $f_0$, and after $t$, the sequence are i.i.d. with pdf $f_1$. Under an average run length (ARL) to false alarm constraint, namely the expected duration to a false alarm is at least $\gamma$,  Lorden's setup is to minimize the ``worst-worst case'' detection delay $\sup_{t} \esssup \mathbb{E}_{t}[(T-t+1)^{+}|X_{1},\ldots,X_{t-1}]$, where $T$ is the stopping time at which an alarm is raised, while Pollak's setup is to minimize the ``worst case'' conditional average detection delay $\sup_{t} \mathbb{E}_{t}[(T-t) | T \geq t]$. Since no prior information about the change point is required, these non-Bayesian setups are very attractive for practical applications.

In the above mentioned classic setups, there is no energy constraint and the sensor can take observations at every time slot. In this paper, we extend Lorden's and Pollak's problems to sensors that are powered by renewable energy. In this case, the energy stored in the sensor is replenished by a random process and consumed by taking %, processing and transmitting
observations. The sensor cannot take observations if there is no energy left. Hence, the sensor cannot take observations at every time instant anymore. The sensor needs to plan its use of power carefully. %Intuitively, the sensor should take less samples when it is unlikely that a change has happened, so that the energy can be saved for future use. %This work is motivated by the increasingly implemented energy harvested wireless sensor networks, which are a promising green solution in wireless communication field \cite{Titicipglu:Misc:10}. These sensors, rather than using the traditional batteries, are driven by the energy derived from the external source, such as solar, mechanical or thermal energy~\cite{Piorno:Proc:09}. Besides their energy being cleaner, the energy harvesters also reduce the cost associated with powering systems and lengthen the life expectancy of wireless sensors. Hence, it is critical to design the optimal quickest detection schemes for the energy harvested sensor network.
Moreover, the stochastic nature of the energy replenishing process will certainly affect the performance of change detection schemes. %One prominent feature of energy harvester is that its performance is dependent on the ambient environment.
Since the energy collected by the harvester in each time instant is not a constant but a random variable, this brings new optimization challenges. %In this paper, we assume the random energy arriving conforms to the Bernoulli distribution with probability $p$ at each time instant. Under this condition, we want to design an optimal power allocation strategy and an optimal detection scheme that jointly minimize the worst case detection delay under certain ARL constraint.

We first consider the scenario in which a unit of energy arrives with probability $p$ at each time instant. For Lorden's setup, two types of ARL constraint are considered in this paper. The first type is an algorithm level ARL constraint, which puts a lower bound on the expected number of observations taken by the sensor before it runs a false alarm. The algorithm level ARL constraint is independent of the energy arriving probability $p$. Under this setup, we prove that the optimal detection procedure is the well known cumulative sum (CUSUM) procedure proposed in \cite{Lorden:AmS:71}, and the optimal power allocation scheme is to allocate the energy as soon as it is harvested. The second type ARL constraint is on a system level, which puts a lower bound on the expected duration to a false alarm. This constraint is related to the energy arriving probability $p$. In this case, we show that CUSUM procedure and the immediate power allocation strategy is asymptotically optimal when the system ARL goes to infinity. For Pollak's setup, we discuss the problem only with the system level ARL in detail. As we can see later, the immediate power allocation coupled with CUSUM detection is actually asymptotically optimal for both the system level ARL and the algorithm level ARL.

We then consider a more general energy arriving process in which more than one unit of energy can arrive at each time instant. In this scenario, we show that a simple energy allocation policy, in which the sensor takes samples as long as there is energy left at the battery, coupled with CUSUM test is asymptotically optimal for both Lorden and Pollak's setups when the system level ARL goes to infinity.

There have been some existing works on the quickest change point detection problem that take the sample cost into consideration. The first main line of existing work considers the problem under a Bayesian setup. The main difference between the Bayesian setup and non-Bayesian setup is that in the Bayesian setup, the change point is modeled as a random variable with a known distribution. No such assumption is made in the non-Bayesian setup. \cite{Premkumar:INFOC:08} considers the design of detection strategy that strikes a balance between the detection delay, false alarm probability and the number of sensors being active. In particular, \cite{Premkumar:INFOC:08} considers a wireless network with multiple sensors monitoring the Bayesian change in the environment. Based on the observations from sensors at each time slot, the fusion center decides how many sensors should be active in the next time slot to save energy.~\cite{Banerjee:SADA:12} take the average number of observations into consideration, and provides the optimal solution along with low-complexity but asymptotically optimal rules. In~\cite{Banerjee:TIT:12}, the authors propose a DE-CUSUM scheme for the non-Bayesian setup and show that it is asymptotically optimal.

The remainder of this paper is organized as follows. The mathematical model is given in Section  \ref{sec:model}. Section \ref{sec:internal} presents the optimal solution for Lorden's problem under the algorithm level ARL constraint and the performance analysis for the optimal solution. In Section \ref{sec:external}, we present asymptotically optimal solutions for Lorden's and Pollak's problems under the system level ARL constraint. Section~\ref{sec:extension} presents our results for a more general energy arriving model. Numerical examples are given in Section \ref{sec:simulation} to illustrate the results obtained in this work. Finally, Section \ref{sec:conclusion} offers concluding remarks.

\section{Problem Formulation} \label{sec:model}
Let $\left\{ X_{k}, k = 1, 2, \ldots \right\}$ be a sequence of random variables whose distribution changes at a fixed but unknown time $t$. Before $t$, the $\left\{ X_{k}\right\}$'s are i.i.d. with pdf $f_{0}$; after $t$, they are i.i.d. with pdf $f_{1}$. The pre-change pdf $f_{0}$ and post-change pdf $f_{1}$ are perfectly known by the sensor. We use $P_{t}$ and $\mathbb{E}_{t}$ to denote the probability measure and the expectation with the change happening at $t$, respectively, and use $P_{\infty}$ and $\mathbb{E}_{\infty}$ to denote the case $t=\infty$.

We assume that the energy arrives randomly at each time slot. To facilitate the presentation and set up notation, we present the model for the case when the energy arriving process is a Bernoulli process with parameter $p$ in this section. A more general model will be considered in Section~\ref{sec:extension}. Specifically, we use $\nu = \{\nu_{1}, \nu_{2}, \dots, \nu_{k}, \dots \}$ to denote the energy arriving process with $\nu_{k} \in \{0, 1\}$, in which $\nu_{k} = 1$ indicates that a unit of energy is collected by the energy harvester at time slot $k$ and $\nu_{k} = 0$ means that no energy is harvested. $\{\nu_{k} \}$ is i.i.d. over $k$. Moreover, we use $P^{\nu}$ to denote its probability measure (correspondingly, we use $\mathbb{E}^{\nu}$ to denote the expectation according to the measure $P^{\nu}$), and we have $P^{\nu}(\nu_{k} = 1) = p$.

The sensor can decide how to allocate these collected energies. Let $\mu = \{\mu_{1}, \mu_{2}, \dots, \mu_{k}, \dots \}$ be the power allocation strategy, where $\mu_{k} \in \{ 0, 1 \}$, in which $\mu_{k} = 1$ means that the wireless sensor spends a unit of energy on taking an observation at time slot $k$, while $\mu_{k} = 0$ means that no energy is spent at time $k$ and hence no observation is taken.

The sensor's battery has a finite capacity $C$. The energy arriving process and the energy utilizing process will affect the amount of energy left in the battery. We use $E_{k}$ to denote the energy left in the battery at the end of time slot $k$. $E_{k}$ evolves according to:
$$ E_{k} = \min[C, E_{k-1} + \nu_{k} - \mu_{k}]. $$

The energy allocation policy $\mu$ must obey the causality constraint, namely the energy cannot be used before it is harvested. The energy causality constraint can be written as
\begin{eqnarray}
E_{k} \geq 0 \quad k = 1,2,\ldots. \label{eq:causality1}
\end{eqnarray}
We use $\mathcal{U}$ to denote the set of all $\mu$'s that satisfy \eqref{eq:causality1}.

The sensor spends energy to take observation. The observation sequence is denoted as \\$\left\{ Z_{k}, k = 1, 2, \ldots \right\}$, where
\begin{eqnarray}
Z_{k}= \left\{ \begin{array}{ll}
X_{k} & \textrm{if  } \mu_{k}=1 \\
\phi & \textrm{if  } \mu_{k}=0
\end{array} \right. .          \label{eq:observation}
\end{eqnarray}
We call an observation $Z_k$ a non-trivial observation if $\mu_k=1$, i.e., if the observation is taken from the environment.

$\{Z_{k}\}$'s are not necessarily conditionally (conditioned on the change point) i.i.d. due to the existence of $\{\mu_{k}\}$. The distribution of $Z_{k}$ is related to both $\mu_{k}$ and $X_{k}$. Therefore, we use $P_{t}^{\mu}$ and $\mathbb{E}_{t}^{\mu}$ to denote the probability measure and expectation of the observation sequence $\{Z_{k}\}$ with the change happening at $t$, respectively.

In this paper, we want to find a stopping time $T$, at which the sensor will declare that a change has occurred, and a power allocation rule $\mu$ that jointly minimize the detection delay. Clearly, the power allocation strategy $\mu_k$ depends causally on the observation process, %$\mathbf{Z}_{1}^{k-1}$
the energy arriving process %$\nu_{1}^{k}$
and the energy utilization process: %$\mu_{1}^{k-1}$:
\begin{eqnarray}
\mu_{k} = g_{k}(\mathbf{Z}_{1}^{k-1}, \nu_{1}^{k}, \mu_{1}^{k-1}), \no
\end{eqnarray}
in which $\mathbf{Z}_{1}^{k-1}$ denotes the vector $[Z_{1},\ldots, Z_{k-1}]$, $\nu_{1}^{k}$ and $\mu_{1}^{k-1}$ are defined similarly, and $g_{k}$ is the power allocation function used at time slot $k$.

We consider three problem setups. The first one is Lorden's quickest detection problem with an  \emph{algorithm level} ARL constraint, which is formulated as
\begin{eqnarray}
\hspace{-12mm} &&\text{(P1)} \quad \quad \min_{\mu \in \mathcal{U}, T \in \mathcal{T}} d(\mu, T), \no \\
\hspace{-12mm} &&\quad \quad \quad \quad \text{s.t.  } \mathbb{E}_{\infty}[N] \geq \eta, \label{eq:prob1}
\end{eqnarray}
where $\mathcal{T}$ is the set of all stopping time with $\mathbb{E}_{t}^{\mu}[T] < \infty$, $N$ is the total number of non-trivial observations taken by the sensor before it claims that the change has happened and
\begin{eqnarray}
&& d(\mu, T) = \sup_{t \geq 1} d_{t}(\mu, T), \no \\
&& d_{t} (\mu, T) = \esssup \mathbb{E}_{t}^{\mu}\left[ (T-t+1)^{+}|\mathcal{F}_{t-1}\right], \label{eq:tdd}
\end{eqnarray}
where $\mathcal{F}_k$ is the set of all observations till time $k$, namely $\mathcal{F}_k = \{Z_1,\cdots,Z_k\}$.
In this case, we put a lower bound $\eta$ on the average number of observations taken before a false alarm is raised. The larger $\eta$ is, the less frequent a false alarm will be raised. Since this constraint is independent of the power allocation scheme $\mu$ and energy arriving sequence $\nu$, this problem setup is robust against the variation of the ambient environment.

The second problem considered in this paper is Lorden's quickest detection problem with a \emph{system level} ARL constraint, which is formulated as
\begin{eqnarray}
\hspace{-12mm} &&\text{(P2)} \quad \quad \min_{\mu \in \mathcal{U}, T \in \mathcal{T}} d(\mu, T), \no \\
\hspace{-12mm} &&\quad \quad \quad \quad \text{s.t.  } \mathbb{E}_{\infty}^{\mu}[T] \geq \gamma. \label{eq:prob2}
\end{eqnarray}
In this formulation, a lower bound is set on the expected duration to a false alarm. In contrast to the previous case, this constraint depends on the power allocation $\mu$, which is further related to the energy arriving probability $p$. Therefore, this setup is more sensitive to the environment. %For a fixed power allocation scheme and a given detection procedure, the algorithm level ARL and the system level ARL are related. We will discuss this relationship in Section \ref{sec:inner_performance} for the optimal solution of (P1).

In some applications, Pollak's formulation is of interest since its delay metric is less conservative than that of Lorden's formulation. In our context, Pollak's formulation can be written as
\begin{eqnarray}
\hspace{-12mm} &&\text{(P3)} \quad \quad \min_{\mu \in \mathcal{U}, T \in \mathcal{T}} \sup_{t \geq 1} \mathbb{E}_{t}^{\mu}\left[ T-t | T \geq t \right], \no \\
\hspace{-12mm} &&\quad \quad \quad \quad \text{s.t.  } \mathbb{E}_{\infty}^{\mu}[T] \geq \gamma. \label{eq:prob3}
\end{eqnarray}
Even without the additional energy casuality constraint, the optimal solution for Pollak's formulation is still unknown. Therefore, in this paper, we discuss only the asymptotic solution for Pollak's formulation. In the sequel, we will see that the proposed asymptotically optimal solution under the system level ARL constraint is also asymptotically optimal under the algorithm level ARL constraint. Hence, in the paper, we discuss only the system level ARL constraint for Pollak's formulation in detail.

For an arbitrary realization of the power allocation scheme $\mu$, we will use the following notation throughout of the paper:
\begin{enumerate}
\item $\{a_{k}, k = 1, 2, \ldots \}$ to denote the time instants at which the energy harvester harvests a unit of energy, i.e., $\nu_{a_{k}} = 1$;
\item $\{b_{k}, k = 1, 2, \ldots \}$ to denote the time instants at which the sensor takes observations, i.e., $\mu_{b_{k}} = 1$;
\item $\left\{X_{k}^{(a_{k},b_{k})}, k = 1, 2, \ldots \right\}$ or $\left\{ \tilde{X}_{k}, k = 1, 2, \ldots \right\}$ to denote the non-trivial observation sequence, which is the subsequence of $\{Z_{k}, k = 1, 2, \ldots\}$ with all its non-trivial elements. In particular, $X_{k}^{(a_{k},b_{k})}$ will be used when we want to emphasize the sampling time. Here $X_{k}^{(a_{k},b_{k})}$ is the $k^{th}$ non-trivial observation taken by the sensor at time $b_{k}$ using the energy arriving at time $a_{k}$.
\end{enumerate}
Using above notation, the energy causality constraint indicates the following inequality:
\begin{eqnarray}
b_{k} \geq a_{k}, \quad k=1, 2, \ldots. \label{eq:causality2}
\end{eqnarray}

An example of the realization of the sensor sampling procedure (and corresponding notation) is shown in Figure \ref{fig:model}.
\begin{figure}[thb]
\centering
\includegraphics[width=0.45 \textwidth]{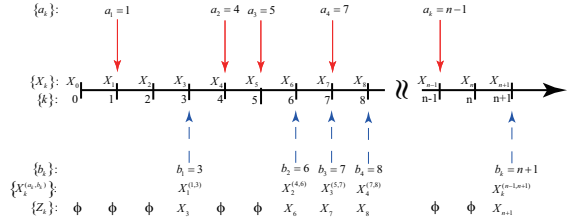}
\caption{An example of the realization of the sampling procedure}
\label{fig:model}
\end{figure}

\section{Optimal solution for Lorden's formulation with the algorithm level ARL constraint}  \label{sec:internal}
%\subsection{Optimal solution} \label{sec:inner_optimal}
In this section, we study the optimal solution for (P1). We use $L(\cdot)$ to denote the likelihood ratio (LR), and use $l(\cdot) = \log L(\cdot)$ to denote the log likelihood ratio (LLR). For the observation sequence $\{Z_{k}\}$, LR is defined as
\begin{eqnarray}
L(Z_{k}) = \left\{ \begin{array}{c c} \frac{f_{1}(Z_{k})}{f_{0}(Z_{k})}, &  \text{if } \mu_{k} = 1 \\
1, & \text{if } \mu_{k} = 0  \end{array} . \right. \label{eq:LR}
\end{eqnarray}
The CUSUM statistic and Page's stopping time can be written as \cite{Lorden:AmS:71}
$$ S_{k} = \max_{1 \geq q \geq k} \left[ \prod_{i=q}^{k}L(Z_{i})\right] = \max[S_{k-1}, 1]L(Z_{k}),$$
and
$$T_{p} = \inf \{ k \geq 0 | S_{k} \geq B\},$$
respectively.

Generally, for a given detection strategy pair $(\mu, T)$, the detection delay $d_{t}(\mu, T)$ in \eqref{eq:tdd} varies from different change point $t$. If there is an equalizer strategy which makes $d_{t}(\mu, T)$ be a constant over $t$, it might be a good candidate for the optimal strategy for the minmax problem. Similar to the conclusion that Page's stopping time is an equalizer rule for the classical Lorden's problem \cite{Poor:Book:08}, we have following proposition:
\begin{prop} \label{prop:equilazer}
The power allocation scheme $\mu^{*} = \nu$ and Page's stopping time $T_{p}$ together achieve an equalizer rule, i.e., $d_{t}(\mu^{*}, T_{p}) = d_{1}( \mu^{*}, T_{p}), \forall t \geq 1$.
\end{prop}
\begin{proof}
Since $\mu^{*} = \nu$ indicates that $\left\{ \mu_{k}^{*} \right\}$'s are i.i.d. over $k$, $\{Z_{k}\}$'s are conditionally i.i.d. given the change point $t$.

Notice that $W_{k} = \max[S_{k}, 1]$ is a non-decreasing function of $S_{k}$, and on the event $\{T_{p} \geq t\}$, $T_{p}$ is a non-increasing function of $W_{t-1}$. Then we have
\begin{eqnarray}
d_{t}(\mu^{*}, T_{p}) &=& \esssup \mathbb{E}_{t}^{\mu^{*}}\left[ T_{p} - t+1 | \mathcal{F}_{t-1}\right] \no \\
&=& \mathbb{E}_{t}^{\mu^{*}}\left[ T_{p} - t+1 | W_{t-1} = 1\right].
\end{eqnarray}
Since $\{W_{k}\}$ is a homogeneous Markov chain under the power allocation scheme $\mu_{k}^{*} = \nu_{k}$, then, $d_{t}(\mu^{*}, T_{p})=d_{1}(\mu^{*}, T_{p})$.
\end{proof}
\begin{rmk}
$\mu^{*} = \nu$ indicates $\mu_{k}^{*} = \nu_{k}$ for every $k$, that is, the sensor spends the energy taking observation immediately when it obtains an energy from the environment. Therefore, we call $\mu^{*}$ the \emph{immediate power allocation scheme} in the sequel.
\end{rmk}

The next lemma shows that the immediate power allocation scheme along with the CUSUM detection scheme is optimal for (P1).
\begin{lem}\label{lem:inner}
The optimal power allocation strategy for (P1) is $\mu^{*}$, and the optimal stopping time is $T_{p}$ with the threshold $B$ being a constant such that $\mathbb{E}_{\infty}[N] = \eta$.
\end{lem}
\begin{proof}
The proof consists of two steps. The first step is to show that for an arbitrary but given power allocation $\mu$, $T_{p}$ is the optimal stopping time. The second step is to show that under $T_{p}$, $\mu^{*}$ is the optimal power allocation scheme. A detailed proof is provided in Appendix \ref{apd:inner}.
\end{proof}

%We denote this optimal pair as $(\mu^{*}, T_{p})$, which indicates that for the Lorden's quickest detection problem with the algorithm level ARL constraint, the optimal power allocation scheme is to spend the energy on taking observation as soon as it is harvested, and the optimal detection scheme is CUSUM.

%\subsection{Performance analysis}\label{sec:inner_performance}
In the following, we analyze the performance of $(\mu^{*}, T_{p})$ %the immediate power allocation scheme and the CUSUM detection strategy
by determining the detection delay and the algorithm level ARL. Since $\{ Z_{k} \}$ is a conditionally i.i.d. sequence under $\mu^{*}$, we can apply Wald's lemma~\cite{Poor:Book:08} in our analysis. We have the following proposition:

\begin{prop} \label{prop:inner_performance}
Suppose $B>1$, then
\begin{eqnarray}
\mathbb{E}_{\infty}[N] &=& \frac{\mathbb{E}_{\infty}[\kappa]}{1-P_{\infty}(F_{0})}, \label{eq:far}\\
d(\mu^{*}, T_{p}) &=& \frac{1}{p} \frac{\mathbb{E}_{1}[\kappa]}{1-P_{1}(F_{0})}, \label{eq:add}
\end{eqnarray}
where $\kappa$ is the stopping time
\begin{eqnarray}
%\kappa &=& \min \left\{ m \geq 1 \Bigg| \sum_{k=1}^{m} l\left(X_{k}^{(a_{k}, a_{k})}\right) \not \in (0, \log B) \right\}, \no
\kappa &=& \min \left\{ m \geq 1 \Bigg| \sum_{k=1}^{m} l\left(\tilde{X}_{k}\right) \not \in (0, \log B) \right\}, \no
\end{eqnarray}
and $F_{0}$ denotes the event
\begin{eqnarray}
%\left\{ \sum_{k=1}^{m} l\left(X_{k}^{(a_{k}, a_{k})}\right) \leq 0 \right\}. \no
\left\{ \sum_{k=1}^{m} l\left(\tilde{X}_{k}\right) \leq 0 \right\}. \no
\end{eqnarray}
\end{prop}
\begin{proof}
The proof follows closely that of Theorem 6.2 in \cite{Poor:Book:08}. A detailed proof is given in Appendix \ref{apd:inner_performance}.
\end{proof}
%
%For a given power allocation scheme and a fixed detection scheme, the algorithm level ARL and the system level ARL are related. The following proposition describes this relationship under the immediate power allocation and CUSUM.
%\begin{prop} \label{prop:in_ext}
%Under the optimal solution, $(\mu^{*}, T_{p})$, for (P1), the system level ARL is given as
%\begin{eqnarray}
%\mathbb{E}_{\infty}^{\mu^{*}}[T_{p}] = \frac{1}{p} \frac{\mathbb{E}_{\infty}[\kappa]}{1-P_{\infty}(F_{0})}
%\end{eqnarray}
%\end{prop}
%\begin{proof}
%Following the similar argument used in Proposition \ref{prop:inner_preformance}, we have
%$$\mathbb{E}_{\infty}^{\mu^{*}}[T_{p}] = \mathbb{E}_{\infty}^{\mu^{*}}[a_{N}] = \mathbb{E}_{\infty}^{\mu^{*}}\left[\sum_{l=1}^{N}\tau_{l}\right] = \frac{1}{p}\mathbb{E}_{\infty}[N].$$
%Then, the conclusion follows from \eqref{eq:far}.
%\end{proof}

We note that in Proposition \ref{prop:inner_performance}, ARL and $d(\mu^{*},T_h^c)$ are given as functions of $P_{\infty}(F_0)$ and $P_{1}(F_0)$, whose precise values are difficult to evaluate. The following result, which is an extension of Lorden's asymptotical result~\cite{Lorden:AmS:71}, shows $d(\mu^{*},T_h^c)$ scales linear with $\log \eta$ when $\eta \rightarrow \infty$.
\begin{prop} \label{prop:inner_asym}
As $\eta \rightarrow \infty$, we have
\begin{eqnarray}
d(\mu^{*}, T_{p}) \sim \frac{1}{p} \frac{|\log \eta|}{I},
\end{eqnarray}
in which $I=I(f_{1}, f_{0})$ is the KL divergence of $f_{1}$ and $f_{0}$.
\end{prop}
\begin{proof}
This statement can be shown by discussing the relationship between one-sided sequential probability ratio test (SPRT) and CUSUM. The discussion is similar to the proof of Lemma \ref{lem:lorden_ext_asym}, therefore, we omit the proof for brevity.
\end{proof}

\section{Asymptotically optimal solution under the system level ARL constraint} \label{sec:external}
In this section, we consider (P2) and (P3). Since both the detection delay and the system level ARL constraint are related to the power allocation $\mu$, it is generally difficult to solve these coupled problems. Inspired by the previous section, we propose to use the simple detection strategy $(\mu^{*}, T_{p})$. We will show that this simple strategy is asymptotically optimal for (P2) and (P3) as $\gamma \rightarrow \infty$. %First, we have:
%\begin{prop}
%$(\mu^{*}, T_{p})$ is asymptotically optimal for (P2) and (P3) when $p \rightarrow 1$.
%\end{prop}
%\begin{proof}
%Under the condition $p=1$, (P2) degenerates to classical Lorden's problem, (P3) degenerates to classical Pollak's problem and $(\mu^{*}, T_{p})$ degenerates to CUSUM procedure. Then, the statement follows immediately.
%\end{proof}

The asymptotic optimality of $(\mu^{*}, T_{p})$ in the rare false alarm region ($\gamma \rightarrow \infty$) can be shown by two steps. In the first step, we derive a lower bound on the detection delay for any power allocation and detection scheme. In the second step, we show that $(\mu^{*}, T_{p})$ achieves this lower bound, which then implies that $(\mu^{*}, T_{p})$ is asymptotically optimal.

%The following lemma shows a lower bound on the detection delay achieved by any  power allocation and detection scheme pair $(\mu,T)$.
The following lemma presents our lower bound on the detection delay.
\begin{lem} \label{lem:ext_lowerbound}
As $\gamma \rightarrow \infty$,
\begin{eqnarray}
&&\hspace{-10mm}\inf\{ d(\mu, T): \mathbb{E}_{\infty}^{\mu}[T] \geq \gamma \} \no \\
&\geq& \inf\left\{ \sup_{t \geq 1} \mathbb{E}_{t}^{\mu}[T - t | T \geq t]: \mathbb{E}_{\infty}^{\mu}[T] \geq \gamma \right\} \no \\
&\geq& \frac{1}{p} \frac{|\log \gamma|}{I}(1+o(1)).
\end{eqnarray}
\end{lem}
\begin{proof}
Please see Appendix \ref{apd:ext_lowerbound}.
\end{proof}

This lower bound $|\log \gamma| (p I)^{-1} (1+o(1))$ can be obtained by $(\mu^{*}, T_{p})$ for both (P2) and (P3). More specifically, we have
\begin{lem} \label{lem:lorden_ext_asym}
$(\mu^{*}, T_{p})$ is asymptotically optimal for (P2) as $\gamma \rightarrow \infty$. Specifically,
\begin{eqnarray}
d(\mu^{*}, T_{p}) \sim \frac{1}{p} \frac{|\log \gamma|}{I}.
\end{eqnarray}
\end{lem}
\begin{proof}
Please see Appendix \ref{apd:lorden_ext_asym}.
\end{proof}

\begin{lem} \label{lem:pollak_ext_asym}
$(\mu^{*}, T_{p})$ is asymptotically optimal for (P3) as $\gamma \rightarrow \infty$. Specifically,
\begin{eqnarray}
\sup_{t \geq 1} \mathbb{E}_{t}^{\mu^{*}}\left[ T_{p}-t | T_{p} \geq t \right] \sim \frac{1}{p} \frac{|\log \gamma|}{I}.
\end{eqnarray}
\end{lem}
\begin{proof}
Please see Appendix \ref{apd:pollak_ext_asym}.
\end{proof}

As we mentioned in Section \ref{sec:model}, although we consider Pollak's formulation only under the system level ARL constraint in detail in this paper, the proposed strategy $(\mu^{*}, T_{p})$ is also asymptotically optimal for the formulation under the algorithm level ARL constraint, which is stated in the following proposition:
\begin{prop}
$(\mu^{*}, T_{p})$ is asymptotically optimal for Pollak's formulation under the algorithm level ARL constraint as $\eta \rightarrow \infty$, and we have
\begin{eqnarray}
\sup_{t \geq 1} \mathbb{E}_{t}^{\mu^{*}}\left[ T_{p}-t | T_{p} \geq t \right] \sim \frac{1}{p} \frac{|\log \eta|}{I}.
\end{eqnarray}
\end{prop}
\begin{proof}
%For any arbitrary but given strategy $(\mu, T)$, the algorithm level ARL can determine system level ARL.
Following the similar argument used in Proposition \ref{prop:inner_performance}, we have
$$\mathbb{E}_{\infty}^{\mu^{*}}[T_{p}] = \mathbb{E}_{\infty}^{\mu^{*}}[a_{N}] = \mathbb{E}_{\infty}^{\mu^{*}}\left[\sum_{l=1}^{N}\tau_{l}\right] = \frac{1}{p}\mathbb{E}_{\infty}[N].$$
That is, under the immediate power allocation $\mu^{*}$, the algorithm level ARL constraint $\mathbb{E}_{\infty}[N] \geq \eta$ can be equivalently converted into a system level ARL constraint $\mathbb{E}_{\infty}^{\mu^{*}}[T_{p}]$. Setting $\gamma = \eta/p$ for a given $p$, $\eta \rightarrow \infty$ is equivalent to $\gamma \rightarrow \infty$. By Lemma \ref{lem:pollak_ext_asym}, $(\mu^{*}, T_{p})$ is asymptotically optimal under the system level ARL constraint, hence it is asymptotically optimal under the algorithm level ARL constraint.
\end{proof}

\section{Extension}\label{sec:extension}
In this section, we extend the original problem setup by assuming that the energy harvester can receive more than one unit energy at each time slot. Specifically, we assume that the energy arriving sequence $\nu=\{\nu_{1}, \ldots, \nu_{k}, \ldots\}$ is i.i.d. over $k$. $\nu_{k} \in \mathcal{V}=\{0, 1, 2, \ldots \}$, in which $\{ \nu_{k} = 0\}$ means that the energy harvester collects nothing at time slot $k$ and $\{ \nu_{k} = i\}$ means that the energy harvester collects $i$ units of energy at time $k$. We use $p_{i}=P^{\nu}(\nu_{k} = i)$ to denote its probability mass function (pmf). Then the energy left in the battery at the end of time slot $k$ is updated by
$$E_{k} = \min \{C, E_{k-1}+\nu_{k}-\mu_{k}\},$$
and the energy causality constraint indicates $E_{k} \geq 0$.

Under this setup, we consider (P2) and (P3). We propose to use a generalized immediate power allocation strategy:
\begin{eqnarray}
\tilde{\mu}_{k}^{*} = \left\{ \begin{array}{ll}
1 & \text{ if } E_{k-1} + \nu_{k} \geq 1 \no\\
0 & \text{ if } E_{k-1} + \nu_{k} = 0
\end{array} \right..
\end{eqnarray}
That is, the sensor keeps taking observations as long as the battery is not empty.

In the following, we show that this generalized immediate power allocation $\tilde{\mu}^{*}$ combined with Page's stopping time $T_{p}$ is asymptotically optimal for (P2) and (P3) in this random energy arriving case. Corresponding to Lemma \ref{lem:ext_lowerbound}, Lemma \ref{lem:lorden_ext_asym} and Lemma \ref{lem:pollak_ext_asym}, we have following two lemmas:
\begin{lem} \label{lem:general_lowerbound}
As $\gamma \rightarrow \infty$,
\begin{eqnarray}
&&\hspace{-10mm}\inf\{ d(\mu, T): \mathbb{E}_{\infty}^{\mu}[T] \geq \gamma \} \no \\
&\geq& \inf\left\{ \sup_{t \geq 1} \mathbb{E}_{t}^{\mu}[T - t | T \geq t]: \mathbb{E}_{\infty}^{\mu}[T] \geq \gamma \right\} \no \\
&\geq& \frac{1}{\tilde{p}} \frac{|\log \gamma|}{I}(1+o(1)),
\end{eqnarray}
where $\tilde{p} \doteq \mathbb{E}^{\nu}[\tilde{\mu}^{*}]$.
\end{lem}
\begin{proof}
Please see Appendix \ref{apd:general_lowerbound}.
\end{proof}

\begin{lem} \label{lem:general_asym}
$(\tilde{\mu}^{*}, T_{p})$ is asymptotically optimal for (P2) and (P3) as $\gamma \rightarrow \infty$. Specifically,
\begin{eqnarray}
d(\tilde{\mu}^{*}, T_{p}) \sim \frac{1}{\tilde{p}} \frac{|\log \gamma|}{I},
\end{eqnarray}
and
\begin{eqnarray}
\sup_{t \geq 1} \mathbb{E}_{t}^{\tilde{\mu}^{*}}[T_{p} - t | T_{p} \geq t] \sim \frac{1}{\tilde{p}} \frac{|\log \gamma|}{I},
\end{eqnarray}
\end{lem}
\begin{proof}
Please see Appendix \ref{apd:general_asym}.
\end{proof}

%\section{Numerical Simulation}
%\input{simulation}
\section{Numerical Simulation} \label{sec:simulation}
In this section, we give some numerical examples to illustrate the analytical results obtained in this paper. In these numerical examples, we assume that the pre-change distribution $f_{0}$ is zero mean Gaussian with variance $\sigma^2$ and the post-change distribution $f_{1}$ is zero mean Gaussian with variance $P+\sigma^2$. In this case, the KL divergence is $I(f_{1}, f_{0}) = \frac{1}{2}\left[ \log \frac{1}{1+P/\sigma^2} + \frac{P}{\sigma^{2}}\right]$, and the signal-to-noise ratio is defined as $SNR = 10 \log P/\sigma^2$.

In the first example, we illustrate the equalizer property of $( \mu^{*}, T_{p})$ under Lorden's formulation. The equalizer property plays a critical role in the performance analysis, since it allows us to study $d(\mu^{*}, T_{p})$ through a relatively simple expression $\mathbb{E}_{1}^{\mu^{*}}[T_{p}]$. In this example, we compare our optimal strategy with a seemingly reasonable strategy: a save-test power allocation scheme combined with CUSUM. The save-test power allocation is a two-threshold strategy: 1) The sensor saves the collected energy for future use if the energy stored in the sensor is less than a threshold $c_{1}$ and the CUSUM statistic is less than threshold $c_{2}$; and 2) the sensor takes observation when either of these two thresholds is exceeded. This rule says that if the CUSUM statistic is low (suggesting that a change has not happened yet) and the energy stored in the sensor is low, the sensor saves its energy. On the other hand, if either the sensor has enough energy, or the CUSUM statistic is high, the sensor should take an observation.
In this simulation, we set $\sigma^2=1$, $SNR = 0dB$, $p=0.5$ and $\gamma = 560$. The simulation result is shown in Figure \ref{fig:equalizer}. In the figure, the blue line with circles is the performance of $(\mu^{*}, T_{p})$, the green dash line with stars is the performance of the save-test power allocation with CUSUM. This simulation confirms our analysis that $(\mu^{*}, T_{p})$ is an equalizer rule, i.e., $d_{1}(\mu^{*}, T_{p})=d_{t}( \mu^{*}, T_{p})$. However, the save-test power allocation scheme along with CUSUM is not an equalizer rule. Actually, in the save-test power allocation scheme, $d_{1}(\mu, T)$ is larger than others. This is due to the fact that in the first time slot, both the CUSUM statistic and the energy stored in the sensor are zero, hence the sensor chooses to store its energy. The sensor will not take observations until the stored energy exceeds $c_{2}$. The duration of this energy collection period is independent of the change point. Then, the worst case happens at $t=1$, and the detection delay caused by the energy collection period is larger than that caused by the immediate power allocation. Since Lorden's performance metric focuses on the worst case, the save-test power allocation is not as good as the immediate power allocation.

\begin{figure}[thb]
\centering
\includegraphics[width=0.45 \textwidth]{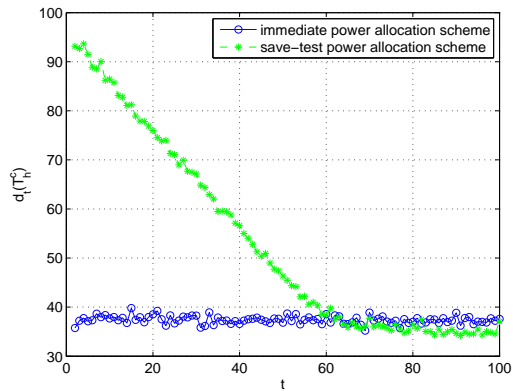}
\caption{The change point $t$ vs $d_{t}(T_{p})$}
\label{fig:equalizer}
\end{figure}

In the second example, we illustrate the relationship between the detection delay and the expected number of observations to false alarm with respect to the energy arriving probability $p$ under setup (P1). In this simulation, we set $\sigma^2 = 1$, $SNR = 0dB$. The simulation result is shown in Figure \ref{fig:ADD_vs_Pfa_1}. In this figure, the blue line with circles is the simulation result for $p=0.2$, the green line with stars and the red line with squares are the results for $p=0.5$ and $p=0.8$, respectively. The black dash line is the performance of the classical Lorden's problem, which serves as a lower bound since in this case the sensor can take observations at every time slot. As we can see, for a given $\eta$, the detection delay is in inverse proportion to the energy arriving probability $p$. The larger $p$ is, the closer is the performance to the lower bound.

\begin{figure}[thb]
\centering
\includegraphics[width=0.45 \textwidth]{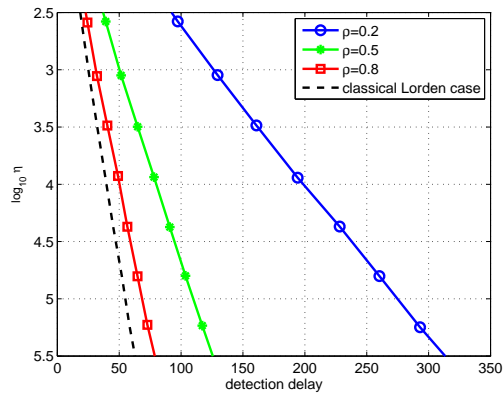}
\caption{Detection delay v.s. the algorithm level ARL}
\label{fig:ADD_vs_Pfa_1}
\end{figure}

\begin{figure}[thb]
\centering
\includegraphics[width=0.45 \textwidth]{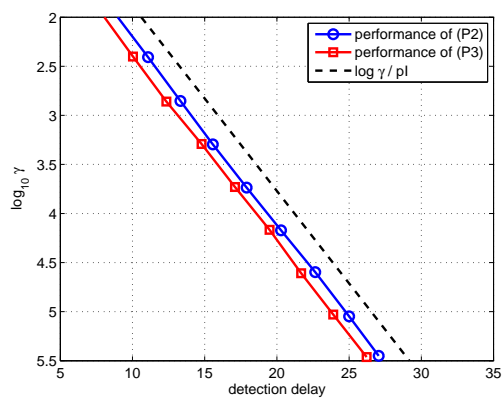}
\caption{Detection delay v.s. the system level ARL}
\label{fig:asym}
\end{figure}

In the third scenario, we examine the asymptotic optimality of $(\mu^{*}, T_{p})$ for (P2) and (P3). In this simulation, we set $p=0.3$, $\sigma^{2}=1$ and $SNR = 5dB$. In this case, we have $I(f_{1}, f_{0}) = 0.8681$. The simulation result is shown in Figure \ref{fig:asym}. In this figure, the blue line with circles is the performance of (P2). The red line with squares is the performance of (P3), and the black dash is calculated by $|\log \gamma|/pI$. Along all the scales, the red curve is below the blue one, which indicates that Pollak's detection delay is smaller than Lorden's detection delay. We also notice that these three curves are parallel to each other, which confirms that the proposed strategy, $(\mu^{*}, T_{p})$, is asymptotically optimal since the difference between them is negligible as $\gamma \rightarrow \infty$.

In the last scenario, we examine the asymptotic optimality of $(\tilde{\mu}^{*}, T_{p})$ for (P2) and (P3) in the extension case that the energy arrives randomly both in amount and in time. In the simulation, we use $C=3$, and we assume that the amount of energy arrives at each time slot takes values in the set $\mathcal{V} = \{0, 1, \ldots, 4\}$. In this case, the probability transition matrix is given as
\begin{eqnarray}
\mathbf{P} = \left[ \begin{array}{c c c c}
p_{0}+p_{1}, &p_{2}, &p_{3}, &p_{4} \\
p_{0}, &p_{1}, &p_{2}, &p_{3}+p_{4} \\
0, &p_{0}, &p_{1}, &\sum_{i=2}^{4} p_{i} \\
0, &0, &p_{0}, &\sum_{i=1}^{4} p_{i} \end{array} \right],
\end{eqnarray}
In the simulation, we set $p_{0}=0.8$, $p_{1}=0.1$, $p_{2}=0.05$, $p_{3}=0.025$, $p_{2}=0.025$, then the stationary distribution is $\tilde{\mathbf{w}}=[0.0182, 0.0545, 0.2000, 0.7273]^{T}$ and $\tilde{p} = 1-p_{0}\tilde{w}_{0} = 0.9964$.

\begin{figure}[thb]
\centering
\includegraphics[width=0.45 \textwidth]{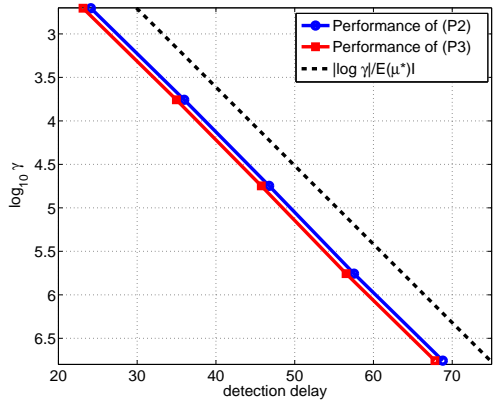}
\caption{Detection delay v.s. the system level ARL}
\label{fig:general_asym}
\end{figure}

In this simulation, we set $\sigma^{2}=1$ and $SNR = 5dB$. The simulation result is shown in Figure \ref{fig:general_asym}. In this figure the blue line with circles is the performance of (P2). The red line with squares is the performance of (P3), and the black dash is calculated by $|\log \gamma|/\tilde{p}I$. Similar to the results obtained in the third simulation scenario, along all the scales, Pollak's detection delay is smaller than Lorden's detection delay, and these three curves are parallel to each other, which confirms that the proposed strategy, $(\tilde{\mu}^{*}, T_{p})$, is asymptotically optimal as $\gamma \rightarrow \infty$.

\section{Conclusion} \label{sec:conclusion}
In this paper, we have studied the non-Bayesian quickest detection problem using a sensor powered by the energy harvested from the environment. Since the energy harvester collected the energy randomly, the quickest detection problem is subjected to a casual energy constraint. Three non-Bayesian quickest detection problem setups, namely Lorden's problem under the algorithm level ARL, Lorden's problem under the system level ARL and Pollak's problem under the system level ARL, have been considered. For the binary energy arriving model, we have shown that the immediate power allocation scheme coupled with CUSUM detection procedure is optimal for the first setup, and is asymptotically optimal for the second and the third setup as ARL goes to infinity. For the more general energy arriving model, we have shown that the proposed generalized immediate power allocation coupled with CUSUM is still asymptotically optimal for the second and third setups.

\appendices
\section{Proof of Lemma \ref{lem:inner}} \label{apd:inner}
We first introduce a notion of quasi change point. For any realization of the power allocation $\mu$, the quasi change point of the non-trivial observation sequence is defined as
\begin{eqnarray}
n = \inf\{k: \tilde{X}_{k} \sim f_{1} \} = \inf \{k: b_{k} \geq t\}. \label{eq:quasi}
\end{eqnarray}
This implies that $n$ can be viewed as the change point happening in the non-trivial observation sequence $\left\{ X_{k}^{(a_{k}, b_{k})} \right\}$. Therefore, a rule minimizing the detection delay $(T-t)^{+}$ among $\{Z_{k}\}$ is the same as the one minimizing $(N-n)^{+}$ among $\left\{ X_{k}^{(a_{k}, b_{k})} \right\}$. Specifically, the stopping rule is decided by
\begin{eqnarray}
&& \min_{N} \sup_{n \geq 1} \esssup \mathbb{E}_{n}\left[ (N-n+1)^{+}|\mathcal{F}_{n-1} \right], \no\\
&& \text{s.t. } \mathbb{E}_{\infty}[N] \geq \eta. \no
\end{eqnarray}
This is the classical Lorden's quickest detection problem \cite{Lorden:AmS:71}, and the optimal solution is given as Page's stopping time $T_{p}$ in \cite{Moustakides:AnS:86} with threshold $B$, which is a constant solely related to $\eta$ and achieves $\mathbb{E}_{\infty}[N] = \eta$.

To prove the optimality of $\mu^{*}$, we examine the following problem:
\begin{eqnarray}
&& \min_{\mu \in \mathcal{U}} \mathbb{E}_{1}^{\mu}[T_{p}] , \no \\
&& \textrm{s.t. } \mathbb{E}_{\infty}[N] = \eta .\label{eq:d1}
\end{eqnarray}
Notice that the objective function is the same as $d_{1}(\mu, T_{p})$. Since
$$\mathbb{E}_{1}^{\mu}[T_{p}] = \mathbb{E}_{1}^{\mu}[b_{N}] \overset{(a)}\geq \mathbb{E}_{1}^{\nu}[a_{N}] \overset{(b)}= \mathbb{E}_{1}^{\mu^{*}}[T_{p}],$$
in which inequality (a) is due to \eqref{eq:causality2}, and equality (b) is true because $T_{p} = a_{N}$ under $\mu^{*} = \nu$. Therefore, $\mu^{*}$ is optimal for the problem \eqref{eq:d1}.

Since
\begin{eqnarray}
\min_{\mu, T} d_{1}(\mu, T) = d_{1}(\mu^{*}, T_{p}) = d_{t}( \mu^{*}, T_{p}), \no
\end{eqnarray}
in which the last equality is due to Proposition \ref{prop:equilazer}, we have
\begin{eqnarray}
d(\mu^{*}, T_{p})=d_{1}(\mu^{*}, T_{p}).  \no
\end{eqnarray}
Combining this with the fact that
\begin{eqnarray}
d(\mu, T) \geq  d_{1}(\mu, T), \no
\end{eqnarray}
we know that $(\mu^{*},T_{p})$ is the optimal solution for (P1).

\section{Proof of Proposition \ref{prop:inner_performance}} \label{apd:inner_performance}
We first examine the quantity $\mathbb{E}_{\infty}[N]$. Consider the non-trivial observation sequence $\left\{X_{k}^{(a_{k}, a_{k})}\right\}$, let $M_{j}$ denote the indicator of the event that the $j^{th}$ repetition of $\kappa$ exits at the upper boundary. That is $M_j=1$ if the $j^{th}$ repetition exits at the upper boundary, and $M_j=0$ if the $j^{th}$ repetition exits at the lower boundary. Let $J$ be a stopping time with respect to the sequence $(\kappa_{1}, M_{1}), (\kappa_{2}, M_{2}), \ldots$, which is i.i.d. under $P_{\infty}$, such that $J=\inf\{j: M_j=1\}$. One can check that $N=\sum_{j=1}^{J} \kappa_{j}$.

From Wald's identity, we have
\begin{eqnarray}
\mathbb{E}_{\infty}[N] = \mathbb{E}_{\infty}\left[ \sum_{j=1}^{J} \kappa_{j}\right] = \mathbb{E}_{\infty}[J]\mathbb{E}_{\infty}[\kappa]. \label{eq:E_N}
\end{eqnarray}

It is easy to see that, under $P_{\infty}$, $J$ is a geometric random variable with
$$P_{\infty}(J=j) = \left[ 1-P_{\infty}(F_{0})\right]\left[ P_{\infty}(F_{0})\right]^{j-1}, \; j=1,2,\ldots.$$
Then, we have
\begin{eqnarray}
\mathbb{E}_{\infty}[J] = \frac{1}{1-P_{\infty}(F_{0})}. \label{eq:E_J}
\end{eqnarray}
Substituting \eqref{eq:E_J} into \eqref{eq:E_N}, we have \eqref{eq:far}.

Following the similar argument as above, we get
$$\mathbb{E}_{1}[N] = \frac{\mathbb{E}_{1}[\kappa]}{1-P_{1}(F_{0})}. $$
Denote $\tau_{i} = a_{i}-a_{i-1}$ as the time interval between two successive observations, the p.m.f. of $\tau_{i}$ is
$$ P(\tau_{i} = j) = (1-p)^{j-1}p,$$
and the average of the time interval between two successive observations is
$$ \mathbb{E}^{\nu}[\tau]  = \frac{1}{p}. $$

For the average detection delay, we have
\begin{eqnarray}
d(\mu^{*}, T_{p}) &=& d_{1}(\mu^{*}, T_{p})\no \\ & =& \mathbb{E}_{1}^{\mu^{*}}[T_{p}] \no \\
&=& \mathbb{E}_{1}^{\mu^{*}}[a_{N}] \no \\ & =& \mathbb{E}_{1}^{\mu^{*}}\left[\sum_{i=1}^{N}\tau_{i}\right] \no \\
&\overset{(a)}=& \mathbb{E}^{\nu}\left[\tau\right]\mathbb{E}_{1}\left[ N \right] \no \\
&=& \frac{1}{p}\mathbb{E}_{1}[N]. \no
\end{eqnarray}
Here, (a) is due to the Wald's identity%, and (b) is due to the independence of $\mu^{*}$ and $\{X_{k}\}$
. Then \eqref{eq:add} follows.

\section{Proof of Lemma \ref{lem:ext_lowerbound}} \label{apd:ext_lowerbound}
This proof relies on several supporting propositions and Theorem 1 of \cite{Lai:TIT:98}.

\begin{prop}\label{prop:convergence}
For an arbitrary but given power allocation $\mu$, we have
\begin{eqnarray}
&&\hspace{-8mm} \lim_{m \rightarrow \infty} \esssup P_{t}^{\mu} \left\{ \frac{1}{m} \max_{0< q \leq m} \sum_{i=t}^{t+q}l(Z_{i}) \geq (1+\varepsilon)I_{1} \Bigg| Z_{1},\ldots,Z_{t-1}\right\} \rightarrow 0  \quad \forall \varepsilon > 0,
\end{eqnarray}
where $I_{1} = pI$.
\end{prop}
\begin{proof}
%%%%%%%%%%%%%%%%%%%% the following is original proof %%%%%%%%%%%%%%%%%%%%%%%%%%%%%%%%%%%%%%%%%%%%%%
We first show that the inequality
\begin{eqnarray}
\frac{1}{m} \sum_{i=t}^{t+m-1} l(Z_{i}) \leq I_{1}, \text{ as } m \rightarrow \infty, \label{eq:asure}
\end{eqnarray}
holds almost surely under $P_{t}^{\mu}$ for any $t \geq 1$.

To show this, we first consider the immediate power allocation $\mu^{*}$, by the strong law of large numbers, we have
\begin{eqnarray}
&&\frac{1}{m} \sum_{i=t}^{t+m-1} \mu_{i} \overset{ a.s.}\rightarrow p, \text{ as } m \rightarrow \infty, \no \\
&&\frac{1}{m} \sum_{i=n}^{n+m-1} l\left(\tilde{X}_{i}\right) \overset{a.s.}\rightarrow I(f_{1},f_{0}), \text{ as } m \rightarrow \infty, \no
\end{eqnarray}
in which $n$ is the quasi change point defined in \eqref{eq:quasi}. Therefore, under $\mu^{*}$, as $ m \rightarrow \infty$, we have
\begin{eqnarray}
\frac{1}{m} \sum_{i=t}^{t+m-1} l(Z_{i}) = \frac{\hat{m}}{m} \frac{1}{\hat{m}} \sum_{i=n}^{n+\hat{m}-1} l\left(\tilde{X}_{i}\right) \overset{a.s.}\rightarrow p I = I_{1}, \label{eq:limit_I}
\end{eqnarray}
where $\hat{m}$ is the number of nonzero elements in $\left\{ \mu_{t}^{*}, \ldots, \mu_{t+m-1}^{*} \right\}$.

For an arbitrary power allocation $\mu$ with $\limsup_{k\rightarrow\infty}\mu_{k}=1$, we always have $\tilde{m} \leq \hat{m}+C$ because of the causal energy constraint, where $\tilde{m}$ denotes the number of nonzero elements in $\left\{ \mu_{t}, \ldots, \mu_{t+m-1} \right\}$. Therefore, as $m \rightarrow \infty$,
\begin{eqnarray}
\frac{1}{m} \sum_{i=t}^{t+m-1} l(Z_{i}) &=& \frac{\tilde{m}}{m} \frac{1}{\tilde{m}} \sum_{i=n}^{n+\tilde{m}-1} l\left(\tilde{X}_{i}\right) \no \\
&\leq& \frac{\hat{m}+C}{m} \frac{1}{\tilde{m}} \sum_{i=n}^{n+\tilde{m}-1} l\left(\tilde{X}_{i}\right) \overset{a.s.} \rightarrow p I. \no
\end{eqnarray}

For the power allocation scheme $\mu$ with $\limsup_{k\rightarrow\infty}\mu_{k}=0$, we have $$\lim_{m\rightarrow\infty}\frac{1}{m} \sum_{i=t}^{t+m-1} l(Z_{i}) = 0 \leq p I.$$
Therefore, inequality \eqref{eq:asure} holds for any arbitrary $\mu$. Notice that i) \eqref{eq:asure} holds in the almost sure sense, since \eqref{eq:limit_I} converges in the almost sure sense; and ii) \eqref{eq:asure} holds for any realization of $Z_{1},\ldots, Z_{t-1}.$

For any $\varepsilon > 0$, define
$$T_{\varepsilon}^{t} = \sup\left\{m\geq 1 \Bigg| \frac{1}{m}\sum_{i=t}^{t+m-1} l(Z_{i}) > (1 + \varepsilon)I_{1} \right\}.$$
Due to \eqref{eq:asure}, we have
$$\essinf P_{t}^{\mu}\{T_{\varepsilon}^{t} < \infty | Z_{1},\ldots,Z_{t-1}\} = 1,$$
which indicates
\begin{eqnarray}
 \lim_{m \rightarrow \infty} \esssup P_{t}^{\mu} \left\{ \frac{1}{m} \max_{0 < q \leq m}  \sum_{i=t}^{t+q} l(Z_{i}) \geq (1+\varepsilon) I_{1} \Bigg| Z_{1},\ldots, Z_{t-1} \right\} \rightarrow 0 \no.
\end{eqnarray}
\end{proof}

Note that Proposition \ref{prop:convergence} holds for every $t \geq 1$, therefore
\begin{eqnarray}
 \lim_{m \rightarrow \infty} \sup_{t \geq 1} \esssup P_{t}^{\mu} \left\{ \frac{1}{m} \max_{0< q \leq m} \sum_{i=t}^{t+q} l(Z_{i}) \geq (1+\varepsilon) I_{1}\Bigg| Z_{1},\ldots, Z_{t-1} \right\} \rightarrow 0. \no\\
\label{eq:laicondition_I}
\end{eqnarray}

To prove Lemma \ref{lem:ext_lowerbound}, we need Theorem 1 in \cite{Lai:TIT:98} , which is restated as follows:
\begin{thm} \label{thm:Lai}
\emph{(\cite{Lai:TIT:98})} Let $\{ Z_{k} \}$ be a random variables sequence with a deterministic but unknown change point $t$. Under probability measure $P_{t}$, the conditional distribution of $Z_{k}$ is $f_{0}(\cdot| \mathbf{Z}_{1}^{k-1})$ for $k<t$ and is $f_{1}(\cdot| \mathbf{Z}_{1}^{k-1})$ for $k \geq t$. Denote $l(Z_{k})$ as
$$ l(Z_{k}) = \log \frac{f_{1}(Z_{k}|\mathbf{Z}_{1}^{k-1})}{f_{0}(Z_{k}|\mathbf{Z}_{1}^{k-1})}.$$
If the condition
\begin{eqnarray}
\lim_{m \rightarrow \infty} \sup_{t \geq 1} \esssup P_{t} \left\{ \max_{0< q \leq m} \sum_{i=t}^{t+q} l(Z_{i}) \geq I_{1}(1+\varepsilon)m \Big| Z_{1},\ldots, Z_{t-1} \right\} \rightarrow 0,  \quad  \forall \varepsilon > 0 \label{eq:necessary_condition}
\end{eqnarray}
holds for some constant $I_{1}$. Then, as $\gamma \rightarrow \infty$,
\begin{eqnarray}
&&\hspace{-6mm}\inf\{ d(\mu, T): \mathbb{E}_{\infty}[T] \geq \gamma \} \no \\
&\geq& \inf\left\{ \sup_{t \geq 1} \mathbb{E}_{t}[T - t | T \geq t]: \mathbb{E}_{\infty}[T] \geq \gamma \right\} \no \\
&\geq& (I_{1}^{-1}+o(1)) \log \gamma. \no
\end{eqnarray}
\end{thm}
\begin{proof}
Please refer to \cite{Lai:TIT:98}.
\end{proof}

In our case, for any arbitrary but given power allocation $\mu$, the conditional density
\begin{eqnarray}
f_{0}^{\mu}( Z_{k} | \mathbf{Z}_{1}^{k-1} )
= f_{0}( X_{k} )P\left(\left\{ \mu_{k}=1  \right\} | \mathbf{Z}_{1}^{k-1} \right) + \delta( \phi )P\left(\left\{ \mu_{k}=0  \right\} | \mathbf{Z}_{1}^{k-1} \right),  \no
\end{eqnarray}
where $\delta( \phi )$ is the Dirac delta function. Similarly, we have
\begin{eqnarray}
f_{1}^{\mu}( Z_{k} | \mathbf{Z}_{1}^{k-1} )
= f_{1}( X_{k} )P\left(\left\{ \mu_{k}=1 \right\} | \mathbf{Z}_{1}^{k-1} \right) + \delta( \phi )P\left(\left\{ \mu_{k}=0 \right\} | \mathbf{Z}_{1}^{k-1} \right).  \no
\end{eqnarray}
Therefore, the log likelihood ratio in Theorem \ref{thm:Lai}
%\begin{eqnarray}
%l(Z_{k}) &=&  \log \frac{f_{1}^{\mu}(Z_{k}|\mathbf{Z}_{1}^{k-1})}{f_{0}^{\mu}(Z_{k}|\mathbf{Z}_{1}^{k-1})} \no \\
%&=& \left\{ \begin{array}{c c} \log \frac{f_{1}(Z_{k})}{f_{0}(Z_{k})}, &  \text{if } \mu_{k} = 1 \\
%0, & \text{if } \mu_{k} = 0  \end{array} , \right. \no
%\end{eqnarray}
$$ l(Z_{k}) =  \log \frac{f_{1}^{\mu}(Z_{k}|\mathbf{Z}_{1}^{k-1})}{f_{0}^{\mu}(Z_{k}|\mathbf{Z}_{1}^{k-1})}
= \left\{ \begin{array}{c c} \log \frac{f_{1}(Z_{k})}{f_{0}(Z_{k})}, &  \text{if } \mu_{k} = 1 \\
0, & \text{if } \mu_{k} = 0  \end{array} , \right. $$
which is consistent with the definition in \eqref{eq:LR}. Moreover, \eqref{eq:laicondition_I} indicates that, for any arbitrary power allocation, \eqref{eq:necessary_condition} holds for the constant $I_{1} = p I$. Therefore, the conclusion in Theorem \ref{thm:Lai} indicates the result for our case:
\begin{eqnarray}
&&\hspace{-6mm}\inf\{ d(\mu, T): \mathbb{E}_{\infty}^{\mu}[T] \geq \gamma \} \no \\
&\geq& \inf\left\{ \sup_{t \geq 1} \mathbb{E}_{t}^{\mu}[T - t | T \geq t]: \mathbb{E}_{\infty}^{\mu}[T] \geq \gamma \right\} \no \\
&\geq& (I_{1}^{-1}+o(1))\log \gamma. \no
\end{eqnarray}

\section{Proof of Lemma \ref{lem:lorden_ext_asym}} \label{apd:lorden_ext_asym}
%In this appendix we prove Lemma \ref{lem:lorden_ext_asym}.

First, a result similar to Proposition \ref{prop:equilazer} still holds in this case. Specifically,
\begin{prop}
$(\mu^{*}, T_{p})$ is an equalizer rule for (P2), i.e., we have $d_{t}( \mu^{*}, T_{p}) = d_{1}( \mu^{*}, T_{p}), \forall t \geq 1$.
\end{prop}
\begin{proof}
This proof is similar to that of Proposition \ref{prop:equilazer}. Hence, we omit the proof for brevity.
\end{proof}

The rest of the proof can be shown by discussing the relationship between CUSUM and one sided SPRT. Denote SPRT statistic as
\begin{eqnarray}
\Lambda_{1:k} = \prod_{i=1}^{k} L(Z_{i}), \label{eq:SPRT}
\end{eqnarray}
and the stopping time as
$$ T_{s,1} = \inf\left\{k \geq 1 | \Lambda_{1:k} \geq B \right\}.$$
Since the CUSUM statistic
$$ S_{k} = \max_{1 \geq q \geq k} \left[ \prod_{i=q}^{k}L(Z_{i})\right] \geq \prod_{i=1}^{k}L(Z_{i}) = \Lambda_{1:k},$$ we always have
$$\mathbb{E}_{1}^{\mu^{*}}[T_{p}] \leq \mathbb{E}_{1}^{\mu^{*}}[T_{s,1}].$$
By the performance of SPRT (Proposition 4.11 in \cite{Poor:Book:08}), we have
$$ \mathbb{E}_{1}^{\mu^{*}}[T_{s, 1}] \sim \frac{|\log \gamma|}{p I}. $$
Noting that $d(\mu^{*}, T_{p}) = d_{1}( \mu^{*}, T_{p}) = \mathbb{E}_{1}^{\mu^{*}}[T_{p}]$ and using Lemma \ref{lem:ext_lowerbound}, we have
$$ d(\mu^{*}, T_{p}) \sim  \frac{1}{p} \frac{|\log \gamma|}{I}.$$

Moreover, by (10) in Theorem 2 of \cite{Lorden:AmS:71}, the threshold $B=\gamma$ will guarantee
$$ \mathbb{E}_{\infty}^{\mu^{*}}[T_{p}] \geq \gamma. $$
The proof is complete.

\section{Proof of Lemma \ref{lem:pollak_ext_asym}} \label{apd:pollak_ext_asym}
We again consider the one sided SPRT with the threshold $B= \gamma$, which will guarantee
$\mathbb{E}_{\infty}^{\mu^{*}}(T_{p}) \geq \gamma.$

Let $T_{s,t}$ denote the stopping time of SPRT starting at time instant $t$, i.e.,
\begin{eqnarray}
T_{s,t} = \inf \left\{ m \geq 1 \Bigg| \prod_{i=t}^{t+m-1} L(Z_{i}) \geq B \right\}, \no
\end{eqnarray}
then Page's stopping time can be written as
\begin{eqnarray}
T_{p} = \inf \left\{ T_{s,t}+t-1 | t = 1,2,\ldots \right\}. \label{eq:SPRT_CUSUM}
\end{eqnarray}

Note that
\begin{eqnarray}
\left\{ T_{p} < t \right\} = \left\{ T_{s, 1} < t \right\} \cup \ldots \cup \left\{ T_{s, t-1} < 1 \right\} \in \mathcal{F}_{t-1},\no
\end{eqnarray}
therefore,
$$\left\{ T_{p} \geq t \right\} \in \mathcal{F}_{t-1}. $$
Then, for an arbitrary $t$,
\begin{eqnarray}
\hspace{-3mm} \mathbb{E}_{t}^{\mu^{*}}\left[ T_{p}-t | T_{p} \geq t \right] &\overset{(a)}\leq& \mathbb{E}_{t}^{\mu^{*}} \left[ T_{s, t}-1 | T_{p} \geq t \right] \no \\
&\overset{(b)}=& \mathbb{E}_{t}^{\mu^{*}} \left[ T_{s, t}\right]-1 \no \\
&\overset{(c)}=&\mathbb{E}_{1}^{\mu^{*}} \left[ T_{s, 1}\right]-1. \no
\end{eqnarray}
Here, (a) is due to \eqref{eq:SPRT_CUSUM}, (b) is due to the fact that $T_{s, t}$ is independent of $\mathcal{F}_{t-1}$, and (c) is true because $\{ Z_{k} \}$'s are conditionally i.i.d. under $\mu^{*}$.

From Appendix~\ref{apd:lorden_ext_asym}, we have
$$ \mathbb{E}_{1}^{\mu^{*}}[T_{s, 1}] \sim \frac{|\log \gamma|}{p I}. $$

Combining this with Lemma \ref{lem:ext_lowerbound}, we have
$$ \sup_{t \geq 1} \mathbb{E}_{t}^{\mu^{*}}\left[ T_{p}-t | T_{p} \geq t \right] \sim \frac{1}{p} \frac{|\log \gamma|}{I}. $$

\section{Proof of Lemma \ref{lem:general_lowerbound}} \label{apd:general_lowerbound}
%In this appendix we prove Lemma \ref{lem:general_lowerbound}.
We first have the following supporting proposition.
\begin{prop}
$\mathbb{E}^{\nu}[\tilde{\mu}^{*}]$ exists, and $0 < \mathbb{E}^{\nu}[\tilde{\mu}^{*}] \leq 1$.
\end{prop}
\begin{proof}
We show that $E_{k}$ is a regular Markov chain with a finite number of states. It is easy to see that $E_{k}$ have only $C+1$ possible states. If at the end of the previous time slot, the battery has zero energy left, then the transition probability is given as
\begin{eqnarray}
&& P^{\nu}(E_{k+1} = 0 | E_{k}=0 ) = p_{0} + p_{1}, \no \\
&& P^{\nu}(E_{k+1} = j-1 | E_{k}=0 ) = p_{j}, \text{ for } 1 < j \leq C \no\\
&& P^{\nu}(E_{k+1} = C | E_{k}=0 ) = \sum_{j=C+1}^{\infty}p_{j}. \no
\end{eqnarray}
If at the end of the previous time slot, the sensor has $i ( 1 \leq i \leq C)$ units of energy left, the transition probability is given as
\begin{eqnarray}
&& P^{\nu}(E_{k+1} = i-1 | E_{k}=i ) = p_{0}, \no \\
&& P^{\nu}(E_{k+1} = i+j-1 | E_{k}=i ) = p_{j}, \text{ for } 1 \leq j \leq C-i \no \\
&& P^{\nu}(E_{k+1} = C | E_{k}=i ) = \sum_{j=C-i+1}^{\infty}p_{j}. \no
\end{eqnarray}
The above transition probability indicates $E_{k}$ is a regular Markov chain. We denote the stationary distribution as $\tilde{\mathbf{w}}=[\tilde{w}_{0}, \tilde{w}_{1}, \ldots, \tilde{w}_{C}]^{T}$, where $\tilde{w}_{i}$ is the stationary probability for the state $E_{k}=i$, then we have
\begin{eqnarray}
\mathbb{E}^{\nu}[\tilde{\mu}^{*}_{k}] &=& P^{\nu}\left[ \tilde{\mu}^{*}_{k}=1 \right] \no \\
&=& 1-P^{\nu}\left[ \tilde{\mu}^{*}_{k}=0 \right] \no \\
&=& 1-P^{\nu}\left[ \nu_{k}=0\right]P^{\nu}\left[E_{k-1}=0\right] \no \\
&=& 1-p_{0}\tilde{w}_{0} \quad \text{ as } k \rightarrow \infty \no
\end{eqnarray}
exists, and $0 \leq \mathbb{E}^{\nu}[\tilde{\mu}^{*}_{k}] \leq 1$.
\end{proof}

We denote $\tilde{p} = \mathbb{E}^{\nu}[\tilde{\mu}^{*}]$. The rest of the proof follows the one in Appendix \ref{apd:ext_lowerbound} by replacing $p$ with $\tilde{p}$.

\section{Proof of Lemma \ref{lem:general_asym}} \label{apd:general_asym}
We first prove the asymptotic optimality of $(\tilde{\mu}^*,T_p)$ for problem (P2). The proof relies on some supporting propositions and Theorem 4 of \cite{Lai:TIT:98}.

\begin{prop} \label{prop:laicondition2}
For the power allocation scheme $\tilde{u}^{*}$, we have
\begin{eqnarray}
\lim_{m \rightarrow \infty} \sup_{k \geq t \geq 1} \esssup P_{t}^{\tilde{\mu}^{*}} \left\{\frac{1}{m} \sum_{i=k}^{k+m} l(Z_{i}) \leq \tilde{p}I-\delta \Bigg| Z_{1}, \ldots, Z_{k-1} \right\} \rightarrow 0 \quad \forall \delta > 0.
\end{eqnarray}
\end{prop}
\begin{proof}
As we have shown in Proposition \ref{prop:convergence}, for any realization of $Z_{1}, \ldots, Z_{k-1}$, and $\forall k \geq t$, under the power allocation scheme $\tilde{\mu}^{*}$, we have
\begin{eqnarray}
\frac{1}{m} \sum_{i=k}^{k+m-1} l(Z_{i}) \overset{a.s.}\rightarrow \tilde{p} I, \quad m \rightarrow \infty. \no
\end{eqnarray}
Then
\begin{eqnarray}
\lim_{m \rightarrow \infty} \esssup P_{t}^{\tilde{\mu}^{*}}\left\{ \Bigg| \frac{1}{m} \sum_{i=k}^{k+m} l(Z_{i}) - \tilde{p}I \Bigg| \geq \delta \Bigg| Z_{1},\ldots,Z_{k-1} \right\} \rightarrow 0 \quad \forall \delta > 0, \no
\end{eqnarray}
for all $k \geq t$. Therefore
\begin{eqnarray}
\lim_{m \rightarrow \infty} \esssup P_{t}^{\tilde{\mu}^{*}}\left\{ \frac{1}{m} \sum_{i=k}^{k+m} l(Z_{i}) \leq \tilde{p}I - \delta \Bigg| Z_{1},\ldots,Z_{k-1} \right\} \rightarrow 0 \no
%&& \quad \quad \quad \quad \quad \quad \quad \quad \quad \quad \quad \quad \quad \quad \quad \quad \forall \delta > 0, \label{eq:pconverge}
\end{eqnarray}
because the above the expression holds for every $k \geq t$. Then the proposition follows.
%\begin{eqnarray}
%&&\hspace{-8mm} \lim_{m \rightarrow \infty} \sup_{k \geq t \geq 1} \esssup P_{t}^{\tilde{\mu}^{*}} \no \\
%&&\left\{\frac{1}{m} \sum_{i=k}^{k+m} l(Z_{i}) \leq \tilde{p}I-\delta \Bigg| Z_{1}, \ldots, Z_{k-1} \right\} \rightarrow 0. \no
%\end{eqnarray}
%because \eqref{eq:pconverge} holds for every $k \geq t$.
\end{proof}

\begin{prop} \label{prop:laicondition3}
Under the power allocation scheme $\tilde{\mu}^{*}$, Page's stopping time $T_{p}$ satisfies
\begin{eqnarray}
\sup_{k \geq 1} P_{\infty}^{\tilde{\mu}^{*}}(k \leq T_{p} < k+m_{\alpha}) \leq \alpha,
\end{eqnarray}
where
\begin{eqnarray}
\liminf \frac{m_{\alpha}}{|\log \alpha|} > (\tilde{p}I)^{-1}, \no
\end{eqnarray}
but
\begin{eqnarray}
\log m_{\alpha} = o(\log \alpha) \text{ as } \alpha \rightarrow 0. \no
\end{eqnarray}
\end{prop}
\begin{proof}
%Note that
%$$T_{p} = \inf\left\{ \hat{k} : \max_{1 \leq q \leq \hat{k}} \left[ \prod_{i=q}^{\hat{k}} L(Z_{i})\right] \geq B \right\},$$
%then,
For any $k$,
\begin{eqnarray}
&&\hspace{-6mm} P_{\infty}^{\tilde{\mu}^{*}}(k \leq T_{p} < k+m_{\alpha}) \no\\
&=& \sum_{\hat{k}=k}^{k+m_{\alpha}-1} P_{\infty}^{\tilde{\mu}^{*}}(T_{p} = \hat{k}) \no\\
&\leq& \sum_{\hat{k}=k}^{k+m_{\alpha}-1}P_{\infty}^{\tilde{\mu}^{*}}\left\{ \prod_{i=\hat{k}-j}^{\hat{k}} L(Z_{i}) \geq B, \exists 0 \leq j \leq \hat{k}-1 \right\} \no\\
&\overset{(a)}=& \sum_{\hat{k}=k}^{k+m_{\alpha}-1}P_{\infty}\left\{ \prod_{i^{\prime}=k^{\prime}-j^{\prime}}^{k^{\prime}} L(\tilde{X}_{i^{\prime}}) \geq B, \exists 0 \leq j^{\prime} \leq k^{\prime}-1 \right\} \no \\
&\overset{(b)}=& \sum_{\hat{k}=k}^{k+m_{\alpha}-1}P_{\infty}\left\{ \prod_{i^{\prime}=1}^{k^{\prime \prime}} L(\tilde{X}_{i^{\prime}}) \geq B, \exists 0 \leq k^{\prime \prime} \leq k^{\prime} \right\} \no \\
&\overset{(c)}\leq& \sum_{\hat{k}=k}^{k+m_{\alpha}-1} \exp(-\log B) \no\\
&=& m_{\alpha} \exp( - \log B). \label{eq:m_h_a}
\end{eqnarray}
Here, (a) is true because the likelihood ratio of $\left\{ Z_{i} \right\}$ and that of $\left\{ \tilde{X}_{i} \right\}$ are the same. Then we substitute $\left\{ Z_{i} \right\}$ with $\left\{ \tilde{X}_{i} \right\}$, and change the probability measure correspondingly. $i^{\prime}, k^{\prime}$ and $j^{\prime}$ are the new indices in $\left\{ \tilde{X}_{i} \right\}$ corresponding to the original $i$, $\hat{k}$ and $j$ in $\left\{ Z_{i}\right\}$. (b) holds because under $P_{\infty}$, $\left\{ \tilde{X}_{i} \right\}$ are i.i.d., then we reverse the sequence. (c) is due to Doob's martingale inequality, since under $P_{\infty}$, $\left\{ L(\tilde{X}_{i}) \right\}$ is a martingale with expectation $1$.

By \eqref{eq:m_h_a}, we can simply choose $m_{\alpha} = |\log \alpha|(\tilde{p}I)^{-1}+\delta$, and choose $B$, the threshold of CUSUM, such that $m_{\alpha} \exp( - \log B) = \alpha$.
\end{proof}

To prove Lemma \ref{lem:general_asym}, we need Theorem 4 ii) of \cite{Lai:TIT:98} , which is restated as follows:
\begin{thm} \label{thm:Lai2}
\emph{(\cite{Lai:TIT:98})} Let $\{ Z_{k} \}$ be a random variables sequence with a deterministic but unknown change point $t$. Under probability measure $P_{t}$, the conditional distribution of $Z_{k}$ is $f_{0}(\cdot| \mathbf{Z}_{1}^{k-1})$ for $k<t$ and is $f_{1}(\cdot| \mathbf{Z}_{1}^{k-1})$ for $k \geq t$. Denote $l(Z_{k})$ as
$$ l(Z_{k}) =\log\frac{f_{1}(Z_{k}|\mathbf{Z}_{1}^{k-1})}{f_{0}(Z_{k}|\mathbf{Z}_{1}^{k-1})}.$$
Denote $e^c$ as the threshold used in Page's stopping time.
Then
\begin{eqnarray}
\mathbb{E}_{\infty}[T_{p}] \geq e^{c}. \no
\end{eqnarray}
Denote $\bar{\mathbb{E}}_{t}(T)$ as Lorden's detection delay, i.e.,
\begin{eqnarray}
\bar{\mathbb{E}}_{t}(T) = \sup_{t\geq1} \esssup \mathbb{E}_{t}\left[(T-t+1)^{+}|Z_{1},\ldots,Z_{t-1}\right].\no
\end{eqnarray}
If $\forall \delta>0$, the condition
\begin{eqnarray}
\lim_{m \rightarrow \infty} \sup_{k \geq t \geq 1} \esssup P_{t}\left\{\frac{1}{m} \sum_{i=k}^{k+m} l(Z_{i}) \leq I_{1}-\delta \Bigg| Z_{1}, \ldots, Z_{k-1} \right\} \rightarrow 0 \no
\end{eqnarray}
holds for some constant $I_{1}$, and as $\alpha \rightarrow 0$, there exists some $m_{\alpha}$ which dependents only on $\alpha$ such that
\begin{eqnarray}
&&\sup_{k \geq 1} P_{\infty}(k \leq T_{p} \leq k+m_{\alpha}) \leq \alpha, \no
\end{eqnarray}
where
\begin{eqnarray}
\liminf \frac{m_{\alpha}}{|\log \alpha|} > I_{1}^{-1},\no
\end{eqnarray}
but,
\begin{eqnarray}
 \log m_{\alpha} = o(\log \alpha) \text{ as } \alpha \rightarrow 0. \no
\end{eqnarray}
Then,
\begin{eqnarray}
\bar{\mathbb{E}}_{t}(T) \leq (I_{1}^{-1} + o(1))c \quad \text{ as } c \rightarrow \infty. \no
\end{eqnarray}
\end{thm}
\begin{proof}
Please refer to \cite{Lai:TIT:98}.
\end{proof}

By Proportion \ref{prop:laicondition2} and \ref{prop:laicondition3}, $(\tilde{\mu}^{*}, T_{p})$ is a strategy that satisfies the conditions in Theorem \ref{thm:Lai2}. Hence, if we choose $c = \log \gamma$ and $I_{1} = \tilde{p}I$ in the theorem, it is easy to verify that
$d(\tilde{\mu}^{*}, T_{p}) \leq ((\tilde{p}I)^{-1} + o(1))|\log \gamma|\text{ with } \mathbb{E}^{\tilde{\mu}^{*}}_{\infty}(T_{p}) \geq \gamma$. Therefore, $(\tilde{\mu}^{*}, T_{p})$ is asymptotically optimal for (P2).

In the rest of this appendix, we show the asymptotic optimality of $(\tilde{\mu}^*,T_p)$ for problem (P3).
\begin{lem}
\begin{eqnarray}
\sup_{t \geq 1} \mathbb{E}_{t}^{\tilde{\mu}^{*}}\left[ T_{p}-t | T_{p} \geq t \right] \sim \frac{1}{\tilde{p}} \frac{|\log \gamma|}{I}.
\end{eqnarray}
\end{lem}
\begin{proof}
Follow the similar argument in Appendx \ref{apd:pollak_ext_asym}, we have
\begin{eqnarray}
\hspace{-3mm} \mathbb{E}_{t}^{\tilde{\mu}^{*}}\left[ T_{p}-t | T_{p} \geq t \right] &\leq& \mathbb{E}_{t}^{\tilde{\mu}^{*}} \left[ T_{s, t}-1 | T_{p} \geq t \right] \no \\
&=& \mathbb{E}_{t}^{\tilde{\mu}^{*}} \left[ T_{s, t} \right]-1.
\end{eqnarray}
We claim that
$$ \mathbb{E}_{t}^{\tilde{\mu}^{*}} \left[ T_{s, t} | E_{t} = i \right] \leq \mathbb{E}_{t}^{\tilde{\mu}^{*}} \left[ T_{s, t} | E_{t} = 0 \right], \text{ for } i = 1, \ldots, C,$$
that is, at the change point $t$, if there is energy left in the battery, the average detection delay tends to be smaller than that of the case with an empty battery. %given the change has happened, the average detection delay with some energy stored in the battery will not be longer than that when battery is empty.
Since $\mathbb{E}_{t}^{\tilde{\mu}^{*}} \left[ T_{s, t} \right] = \mathbb{E}_{t}^{\tilde{\mu}^{*}} \left[\mathbb{E}_{t}^{\tilde{\mu}^{*}} \left[ T_{s, t} | E_{t} \right]\right]$, we have
$$ \mathbb{E}_{t}^{\tilde{\mu}^{*}}\left[ T_{p}-t | T_{p} \geq t \right] \leq \mathbb{E}_{t}^{\tilde{\mu}^{*}} \left[ T_{s, t} | E_{t} = 0 \right] - 1. $$

Let $B=\gamma$, we have
$$ T_{s,t} = \inf \left\{ m \geq 1 \Bigg| \sum_{i=t}^{t+m} l(Z_{i}) \geq \log \gamma \right\}. $$

We define a sequence of stopping times $\{ T_{s,t}^{(1)}, \ldots, T_{s,t}^{(n)}, \ldots \}$ in the following manner:
\begin{enumerate}
\item Set $E_{t} = 0$. Define
$$T_{s,t}^{(1)} = \inf \left\{ m \geq 1 \Bigg| \sum_{i=t}^{t+m} l(Z_{i}) \geq \log \gamma \right\}.$$
\item Set $E_{T_{s,t}^{(n-1)}} = 0$. Define
$$T_{s,t}^{(n)} = \inf \left\{ m \geq 1 \Bigg| \sum_{i=T_{s,t}^{(n-1)}+1}^{T_{s,t}^{(n-1)}+m} l(Z_{i}) \geq \log \gamma \right\}.$$
\end{enumerate}
That is, at change point $t$, we discard all the energy left in the battery and then start a new SPRT under the power allocation $\tilde{\mu}^{*}$. When the previous SPRT stops, we empty the battery again, and start a new SPRT immediately. Then, this sequence of stopping time $\{T_{s, t}^{(1)}, \ldots, T_{s,t}^{(n)},\ldots,\}$ are independent with the same distribution of $T_{s,t}$ under $E_{t} = 0$. Therefore, by the strong LLN, for an $N$ that large enough, we have
\begin{eqnarray}
\frac{M}{N} = \frac{T_{t}^{(1)} + T_{t}^{(2)} + \dots + T_{t}^{(N)}}{N} \overset{a.s.}\rightarrow \mathbb{E}_{t}^{\tilde{\mu}^{*}}[T_{s,t}|E_{t}=0], \no
\end{eqnarray}
where $M = \sum_{i=1}^{N} T_{s, t}^{(i)}$. Since we have
\begin{eqnarray}
\sum_{i=t}^{t+M} l(Z_{i}) \geq N \log \gamma, \no
\end{eqnarray}
as $\gamma \rightarrow \infty$, $M \rightarrow \infty$, then
\begin{eqnarray}
\frac{1}{M} \sum_{i=t}^{t+M} l(Z_{i}) \geq \frac{N}{M} \log \gamma, \no
\end{eqnarray}
that is
\begin{eqnarray}
\tilde{p}{I} \geq \frac{N}{M} \log \gamma \text{  or  } \frac{M}{N} \geq \frac{|\log \gamma|}{\tilde{p}I}.  \no
\end{eqnarray}
If we ignore the overshoot, we will have
$$\mathbb{E}_{t}^{\tilde{\mu}^{*}}[T_{s,t}|E_{t}=0] \sim \frac{|\log \gamma|}{\tilde{p}I}. $$
Then, we have
\begin{eqnarray}
\mathbb{E}_{t}^{\tilde{\mu}^{*}}\left[ T_{p}-t | T_{p} \geq t \right] \leq \frac{|\log \gamma|}{\tilde{p}I}(1+o(1)). \no
\end{eqnarray}

\end{proof}

\bibliographystyle{ieeetr}{}
\bibliography{macros,detection,energyharvester}

\end{document}